\documentclass[11pt]{article}
\usepackage[utf8]{inputenc}
\usepackage[T1]{fontenc}
\usepackage{graphicx}
\usepackage{placeins}

\usepackage{amsmath,amssymb,amsthm,thmtools}
\declaretheorem{theorem}
\declaretheorem[sibling=theorem]{lemma}
\declaretheorem[sibling=theorem]{proposition}

\usepackage{csquotes}
\usepackage[%
  backend=biber,
  style=numeric-comp,
  url=false,
  giveninits=true,
  maxnames=6
]{biblatex}
\AtBeginBibliography{\scriptsize}

\usepackage{booktabs}
\usepackage[font=scriptsize]{caption}
\usepackage{hyperref}
\usepackage[noabbrev,capitalise]{cleveref}

\author{%
  Christof K\"ulske\footnote{
    Ruhr-Universit\"at Bochum, Fakult\"at f\"ur Mathematik,
    Universit\"atsstra\ss e 150, 44780 Bochum, Germany.  E-mail:
    \texttt{daniel.meissner-i4k@ruhr-uni-bochum.de},
    \texttt{christof.kuelske@ruhr-uni-bochum.de}
  }
  \and
  Daniel Mei\ss ner\footnotemark[\value{footnote}]
}

\title{ Stable and metastable phases for the Curie-Weiss-Potts model
  in vector-valued fields via singularity theory }

\date{\today}

\begin{document}

\maketitle

\begin{abstract}
  We study the metastable minima of the Curie-Weiss Potts model with
  three states, as a function of the inverse temperature, and for
  arbitrary vector-valued external fields. Extending the classic work
  of \textcite{elliswang90} and \textcite{wang94} we use singularity
  theory to provide the global structure of metastable (or local)
  minima. In particular, we show that the free energy has up to four
  local minimizers (some of which may at the same time be global) and
  describe the bifurcation geometry of their transitions under
  variation of the parameters.
\end{abstract}


\section{Introduction}

\subsection{Research context}

The Potts model \parencite{wu82_potts_model}, and in particular its
Curie-Weiss version, is next to the Curie-Weiss Ising model, one of
the most studied models in statistical mechanics.  While basic aspects
of Curie-Weiss models can be discovered by ad-hoc computations, they
provide ongoing challenges for refined problems involving dynamics,
metastability, complex parameters, fine asymptotics \parencite[see for
example][]{Cerf_2016, fernandez12, chatterjee11, gheissari18,
  shamis17, Eichelsbacher_2015, landim16_metas_non_rever_mean_field}.
In particular, motivated by metastability one aims at a full
understanding of the free energy landscape \parencite{Olivieri_2005,
  Bovier_2015}.  The phase-diagram for the stable states of the
Curie-Weiss Potts model, that is the behaviour of \emph{global}
minimizers, is known and described by the Ellis-Wang theorem
\parencite{elliswang90} in zero external field.  \Textcite{wang94}
provides some results in non-zero external field relating to global
minimizers.

For the \emph{metastable} states, that is, the local minima and their
transitions a complete exposition of the global analysis in the whole
parameter space, revealing the full structure of their transitions is
lacking in the literature.  See
\parencite{gaite98,gaite99,gaite_phys_rev_A} for a partial analysis
based on polynomial equivalence given for some regions in parameter
space.

One is not just interested in the static behaviour of the model, but
also in the behaviour under stochastic dynamics, in and out of
equilibrium, where we see next to \emph{metastability} also the
phenomenon of dynamical \emph{Gibbs\,--\,non-Gibbs transitions}
\parencite[see][]{van_Enter_2010_Kuelske_Opoku_Ruszel, Kissel_2020,
  EnFeHoRe, van_Enter_Ermolaev_Iacobelli_Kuelske_2012, Jahnel_2017}.
Employing the probabilistic notion of \emph{sequential Gibbs
  property}, such dynamical Gibbs\,--\,non-Gibbs transitions have been
shown to occur and analyzed for a number of exchangeable models in
\cite{kln2007,kissel18,KRvZ}. A way to understand these transitions
for independent dynamics is to examine the structure of stationary
points of a \emph{conditional rate function} which generalizes the
equilibrium rate function but contains more information.  This
conditional rate function depends on all model parameters of the
static model, but has an additional time-parameter, and a
measure-valued parameter with the meaning of an empirical
distribution. We are planning to investigate the time-evolved
Curie-Weiss Potts model in the spirit of
\cite{kln2007,fernandez12,hollander15}. For this the analysis of the
static problem for \emph{general parameter values (which means for
  general vector-valued fields)} is a necessary first step.

As a guiding principle for such analysis, singularity theory is very
useful to understand the organization of stationary points for varying
parameters.  It allows to understand and discover the types of local
bifurcations which are present in the applied problem as (by Thom's
theorem) they must be related via local diffeomorphisms to (partial)
unfoldings of elementary singularities (so-called catastrophes).
These catastrophes form prototypical types which in themselves can be
most easily understood in polynomial models.  Clear expositions
describing their geometries are found in textbooks
\parencite{poston78,arnold85,lu_intro_sing}.  The simplest example is
the rate function of the Curie-Weiss Ising model.  It is even globally
(for all values of inverse temperature and magnetic field, and all
magnetizations) identical to one cusp singularity.  As we will see,
for the Potts model in a general field various triple points (where
three local minima have merged) occur, and singularity theory becomes
very useful if we want to understand the \emph{global picture of all
  the transitions the metastable minima undergo} for general values of
inverse temperature and fields which allow to lift all degeneracies.

The paper is organized as follows.  As our main result we describe the
geometry of the bifurcation set in the parameter space of inverse
temperature and vector-valued fields (see
Section~\ref{sec:extended-diagram}).  This decomposition of parameter
space given by the bifurcation set provides a phase diagram describing
also the metastable minima.  We will show how some of the so-called
elementary singularities (threefold symmetric butterflies and an
elliptic umbilic surrounded by three folds) and their partial
unfoldings beautifully interact and are glued together.  We also
describe the accompanying geometry of the free energy landscape for
each connected component of the complement of the bifurcation set.
Finally we complement the study with the description of the
coexistence sets of lowest minima (Maxwell-sets) at which two, three,
or four minima coexist (see Section~\ref{sec:classical-diagram}).

\subsection{Model}

We consider the mean-field Potts model with three states and are
interested in both the stable and metastable phases.  Note that we
sometimes use the term local minima to include both the global and
local minima depending on the context.  The space of configurations in
finite-volume \(n \ge 2\) is defined as \(\Omega_n = \{1, 2,
3\}^n\). We define
\begin{equation}
  \Delta^2 = \{\nu \in \mathbb{R}^3 \:|\: \nu_i \ge 0, \sum_{i=1}^3 \nu_i = 1\}
\end{equation}
and often refer to it as the \emph{unit simplex}.  The Hamiltonian of
this model is given by
\begin{equation}
  H_n(\sigma) = -\frac{1}{2n} \sum_{i, j=1}^n \delta_{\sigma_i, \sigma_j}
\end{equation}
where \(\sigma\) lies in \(\Omega_n\).  The vector-valued fields are
equivalently described by the a-priori measure \(\alpha\) in the unit
simplex \(\Delta^2\).  The finite-volume Gibbs distributions are
therefore given by
\begin{equation}
  \mu_{n, \beta, \alpha}(\sigma) = \frac{1}{Z_n(\beta, \alpha)}
  e^{-\beta H_n(\sigma)} \prod_{i=1}^n \alpha(\sigma_i)
\end{equation}
where the partition function is defined as
\begin{equation}
  Z_n(\beta, \alpha) =
  \sum_{\sigma\in\Omega_n}  e^{-\beta H_n(\sigma)} \prod_{i=1}^n \alpha(\sigma_i).
\end{equation}
So the associated free energy in terms of the empirical spin
distribution~\(\nu\) is given by
\begin{equation}
  f_{\beta,\alpha}(\nu) =
  -\frac{1}{2}\beta \langle \nu, \nu \rangle + \sum_{i=1}^3 \nu_i \log \frac{\nu_i}{\alpha_i}.
\end{equation}
This is a real-valued function on the unit simplex \(\Delta^2\) with
two parameters: the inverse temperature \(\beta\) and the external
fields modeled by the a-priori measure \(\alpha\).  Thus, the
parameter space is the product \((0, \infty) \times \Delta^2\).  By a
\emph{phase} (stable or metastable) we mean a (global or local)
minimizer of the free energy \(f_{\beta,\alpha}\).

\begin{figure}
  \centering
  \includegraphics{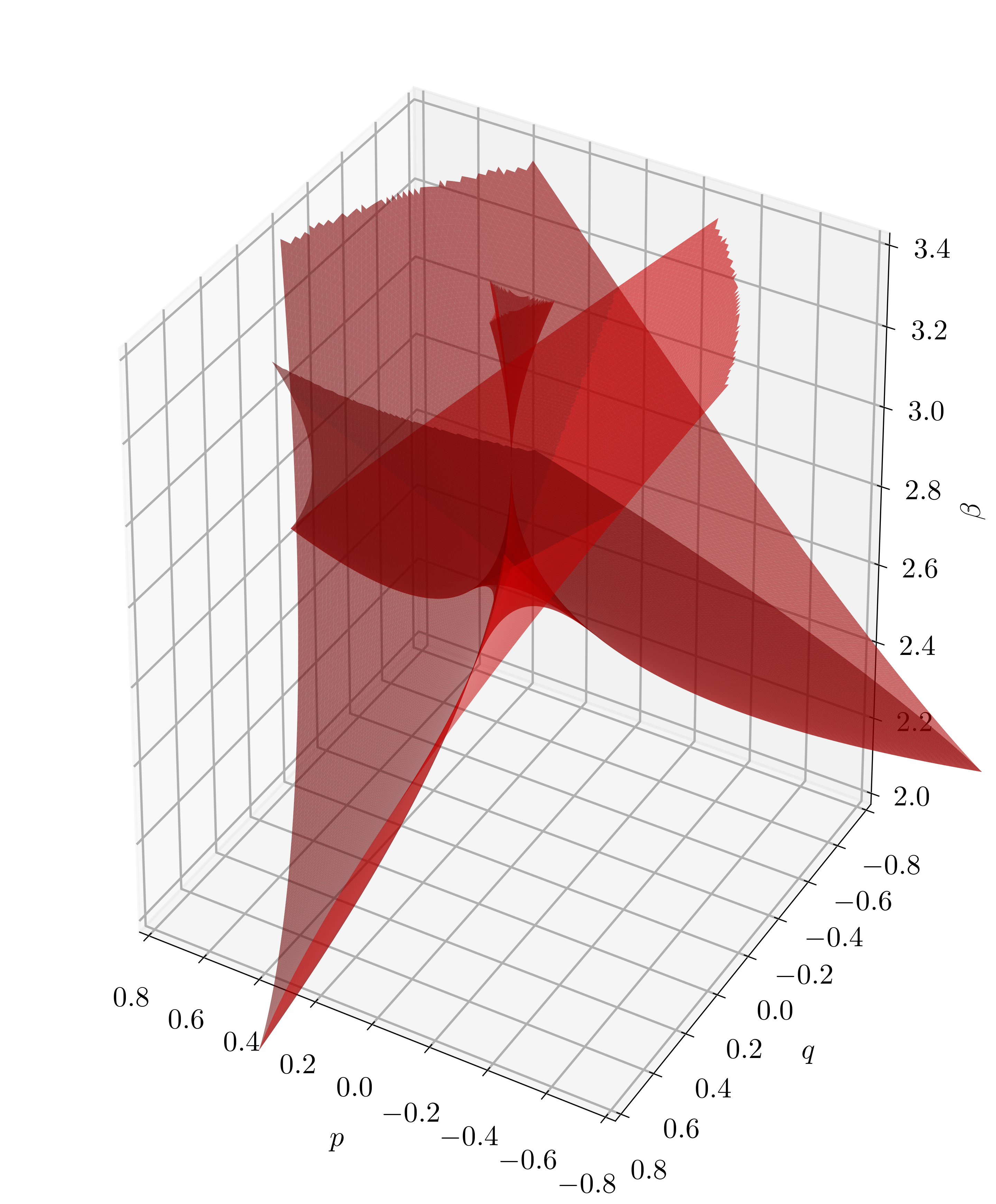}
  \caption{ This plot shows the bifurcation set which is the basis for
    the metastable phase diagram.  The surface shows pinches and
    self-intersections.  Inside of the connected components of the
    complement the structure of metastable phases does not change.
    The unit simplex of a-priori measures \(\alpha\) is embedded in
    the horizontal plane using the \((p, q)\)-coordinates defined in
    \eqref{eq:sym-fields-coords}.}
  \label{fig:3d-bifu-set}
\end{figure}

\section{The metastable phase diagram}
\label{sec:extended-diagram}

In contrast to the usual question of phase-coexistence which is
answered by the stable phase diagram (see
Section~\ref{sec:classical-diagram}), the metastable phase diagram
contains information about the metastable phases of the system but not
of their relative depth.  Mathematically speaking, the
\emph{metastable phase diagram} is a partition of the parameter space
whose cells contain parameter values \((\beta, \alpha)\) such that
\(f_{\beta, \alpha}\) has the same number of local minima.  Using
singularity theory we find that the metastable phase diagram is given
by the connected complements of the surface shown in Figure
\ref{fig:3d-bifu-set}.  The structure of this union is particularly
interesting because it shows features of two well-known catastrophes
\parencite{poston78}: the butterfly catastrophe and the elliptic
umbilic.  The elliptic umbilic permits the change from minimum to
maximum at the centre and is inherently associated to the Potts model
because of its symmetry.  The appearance of the butterfly is connected
to triple points of Ising-like subsystems of the three-state Potts
model.  If we disfavor one of the three states, the remaining two act
similarly to an Ising model in a random field.  There is an
interesting global interdependence of these two different
catastrophes.  We summarize the geometry of the extended phase diagram
in the following theorem.

\begin{theorem}
  For each positive \(\beta\) define the so-called catastrophe map
  \(\chi_{\beta}\) from \(\Delta^2\) to \(\Delta^2\) which associates
  to each empirical spin distribution a \(\beta\)-dependent a-priori
  measure \(\alpha\) modelling the external fields:
  \begin{equation}
    \label{eq:cata-map}
    \chi_{\beta}(\nu) =
    \left( \frac{\nu_i e^{-\beta \nu_i}}{\sum_{k=1}^3 \nu_k e^{-\beta \nu_k}} \right)_{i=1}^3
  \end{equation}
  Define also curves \(\gamma_{\beta}\) via
  \begin{equation}
    \label{eq:gamma-curves}
    \gamma_{\beta}(x) =
    \frac{1}{2}\left( 1 - x - \sqrt{(1 - x)^2 - \frac{4(1 - 2\beta x(1-x))}{\beta(2 - 3\beta x)}} \right)
  \end{equation}
  for \(x \in D_{\beta}\) where the domain \(D_{\beta}\) is a union of
  intervals given by
  \begin{equation}
    \label{eq:domain-gamma-curves}
    D_{\beta} = \begin{cases}
      \left( 0,
        1 - \frac{2}{\beta}
      \right] \cup
      \left( \frac{1}{2} - \frac{1}{2}\sqrt{1 - \frac{2}{\beta}},
        \frac{1}{2} + \frac{1}{2}\sqrt{1 - \frac{2}{\beta}}
      \right) & \text{if } 2 < \beta \le \frac{8}{3},
      \\[9pt]
      \left( 0,
        \frac{1}{2} - \frac{1}{2}\sqrt{1 - \frac{2}{\beta}}
      \right) \cup
      \left( 1 - \frac{2}{\beta},
        \frac{1}{2} - \frac{1}{2}\sqrt{1 - \frac{8}{3\beta}}
      \right)
      \\
      \cup \left( \frac{1}{2} + \frac{1}{2} \sqrt{1 - \frac{8}{3\beta}},
        \frac{1}{2} + \frac{1}{2}\sqrt{1 - \frac{2}{\beta}}
      \right)
      & \text{if } \frac{8}{3} \le \beta < 3,
      \\[9pt]
      \left( 0,
        \frac{1}{2} - \frac{1}{2}\sqrt{1 - \frac{2}{\beta}}
      \right) \cup
      \left( \frac{1}{2} - \frac{1}{2}\sqrt{1 - \frac{8}{3\beta}},
        1 - \frac{2}{\beta},
      \right)
      \\
      \cup \left( \frac{1}{2} + \frac{1}{2}\sqrt{1 - \frac{8}{3\beta}},
        \frac{1}{2} + \frac{1}{2}\sqrt{1 - \frac{2}{\beta}}
      \right)
      & \text{if } 3 \le \beta.
    \end{cases}
  \end{equation}
  Then consider the curve \(\Gamma_{\beta}\) in \(\Delta^2\) given by
  \(\Gamma_{\beta}(x) = \chi_{\beta}(x, \gamma_{\beta}(x), 1 - x -
  \gamma_{\beta}(x))\) with \(x \in D_{\beta}\).  By
  \(S_3\Gamma_{\beta}(D_{\beta})\) we denote the orbit of the curve
  \(\Gamma_{\beta}\) under the action of the permutation group \(S_3\)
  acting on \(\Delta^2\).
  \begin{enumerate}
  \item The constant-temperature slices of the bifurcation set from
    Figure~\ref{fig:3d-bifu-set} are given by
    \(S_3\Gamma_{\beta}(D_{\beta})\).
  \item Table~\ref{tab:slices-overview} shows an overview of the
    number of connected components of \(S_3\Gamma_{\beta}(D_\beta)\)
    together with the possible number of local minima.  The exact
    number of local minima in the respective components can be seen in
    Figure \ref{fig:bifu-slice}.
  \end{enumerate}
\end{theorem}

\begin{table}
  \centering
  \begin{tabular}{rcc}
    \toprule
    & cells of \((S_3\Gamma_{\beta}(D_{\beta}))^\complement\) & number of local minima \\
    \midrule
    \(\beta \le 2\) & 1 & 1\\
    \(2 < \beta \le \frac{18}{7}\) & 4 & 1, 2\\
    \(\frac{18}{7} < \beta < \frac{8}{3}\) & 13 & 1, 2, 3\\
    \(\beta = \frac{8}{3}\) & 16 & 1, 2, 3\\
    \(\frac{8}{3} < \beta < \beta_{\mathrm{cross}}\) & 13 & 1, 2, 3\\
    \(\beta = \beta_{\mathrm{cross}}\) & 12 & 1, 2, 3 \\
    \(\beta_{\mathrm{cross}} < \beta < \beta_{\mathrm{touch}}\) & 13 & 1, 2, 3, 4\\
    \(\beta = \beta_{\mathrm{touch}}\) & 10 & 1, 2, 3, 4 \\
    \(\beta_{\mathrm{touch}} < \beta < 3\) & 8 & 1, 2, 3\\
    \(\beta = 3\) & 7 & 1, 2, 3 \\
    \(3 < \beta\) & 8 & 1, 2, 3 \\
    \bottomrule
  \end{tabular}
  \caption{ Overview of the different \(\beta\)-regimes for the
    constant-temperature slices of the bifurcation set.  The critical
    inverse temperatures \(\beta_{\mathrm{cross}}\) and
    \(\beta_{\mathrm{touch}}\) are defined in
    Subsections~\ref{sec:crossing-temp} and
    \ref{sec:triangle-touch-temp}.}
  \label{tab:slices-overview}
\end{table}

\begin{figure}
  \centering
  \includegraphics{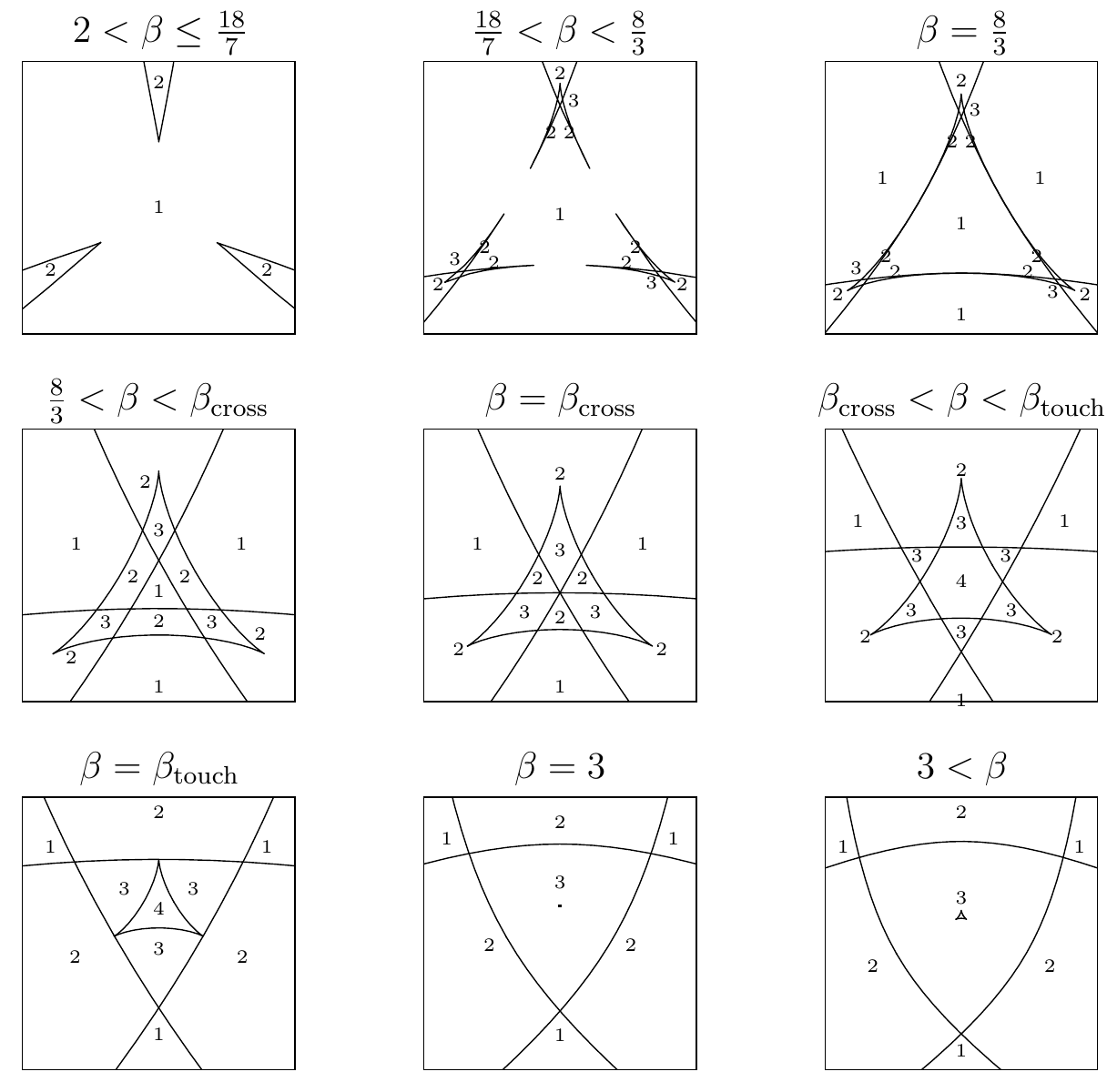}
  \caption{Representative slices through the bifurcation set at
    constant temperatures as indicated by the plot titles.  The
    numbers of the cells show the number of local minima that the rate
    function \(f_{\beta,\alpha}\) has inside of the respective cell.
    The slices are given in
    \((p, q)\)-coordinates~\eqref{eq:sym-fields-coords} as in
    Figure~\ref{fig:3d-bifu-set}.}
  \label{fig:bifu-slice}
\end{figure}

\subsection{Main transitions}

Increasing the inverse temperature from zero we see the following
transitions for slices of the bifurcation set at fixed inverse
temperature, as plots of curves in two-dimensional magnetic field
space (see \cref{fig:bifu-slice}). In connected complements the
structure of stationary points does not change.  We will keep track of
the number of minimizers, which takes all values between one and four.
The Maxwell sets where non-uniqueness of global minimizers occurs are
described in Section~\ref{sec:classical-diagram}.  They have the
meaning of \emph{special magnetic fields} where (in general) the
minimizer can be made to jump by infinitesimal perturbation.  This is
analogous to the notion of \emph{bad empirical measures} in the
dynamical model, in the sense of sequential Gibbsianness
\parencite{kln2007}.

\begin{itemize}
\item \(2 < \beta < \frac{18}{7}\).  First \emph{three symmetric
    cusps} appear at a positive distance to the origin in magnetic
  field space (three \enquote{rockets} pointing towards the origin).
  For magnetic fields inside the cusps we see precisely two minima,
  outside there is one minimum.  For each such inverse temperature,
  the effect of the two-dimensional magnetic field for values in the
  interior of this region to this effective orthogonal Ising model
  translates into an effective inverse temperature times effective
  magnetic field.

\item \(\beta = \frac{18}{7}\) (butterfly).  The three cusps each
  individually develop a butterfly singularity.  The unfolding of the
  pentagram-shaped curve is studied via a Taylor expansion in
  Subsection~\ref{sec:butterfly-taylor}.

\item \(\frac{18}{7} < \beta < \frac{8}{3}\).  The butterfly
  (partially) unfolds, keeping the reflection symmetry.  This
  phenomenon is also known to occur in the Curie-Weiss random field
  Ising model with bimodal disorder \parencite[compare][]{kln2007}.
  The potential has two minima in the outer horns of the pentagram,
  and three minima in the inner horn, as known from the
  one-dimensional polynomial model.  For zero magnetic field there is
  still one minimum in the centre of the simplex.

\item \(\beta = \frac{8}{3}\).  The outer horns (or beaks) of the
  pentagrams grow until they meet symmetrically in a
  \emph{beak-to-beak} singularity. This occurs in three pairs.  A
  one-pair beak-to-beak singularity is also known to occur in the
  parabolic umbilic \parencite[see][]{broecker_1975}.  This touching
  creates a finite connected component at the origin in magnetic field
  space, still with one minimum.

\item \(\frac{8}{3} < \beta < \beta_{\mathrm{cross}} \approx 2.7456\).
  Each two of the beaks (outer horns corresponding to different
  butterflies) have now become connected.  Each such pair now forms a
  joint connected component with two minima.  Their three outer
  boundary curves now form a triangle in the centre.  Each two of the
  former unbounded curves of the butterfly catastrophe have merged
  into one doubly infinite curve at which a fold occurs (fold lines).
  In the connected component at the origin there is still one minimum.

\item \(\beta = \beta_{\mathrm{cross}}\). The triangle stays, the
  three symmetric fold lines move towards the origin.  They pass the
  origin at \(\beta_{\mathrm{cross}}\), when the \enquote{tops of the
    rockets} meet at the origin, and the connected component
  containing the origin with one minimum vanishes.
  \(\beta_{\mathrm{cross}}\) is the parameter value for the appearance
  of symmetric minima near the corners in zero magnetic field.  Hence
  it is simply found by looking at the potential in zero field, along
  the axis of symmetry (see \ref{sec:crossing-temp}).

\item
  \(\beta_{\mathrm{cross}}< \beta < \beta_{\mathrm{touch}}\approx
  2.8024\).  The three rockets move on beyond the origin, they
  intersect, with the \emph{appearance of a middle hexagon}.  In this
  middle hexagon containing the origin there are all four minima
  present.  In zero field the middle minimum is the lowest first but
  moving beyond the Ellis-Wang critical inverse temperature
  \(4 \log 2\), eventually the outer minima become lower. In the
  adjacent six triangles there are three minima.

\item \(\beta \ge \beta_{\mathrm{touch}}\).  Three components with
  three minima vanish, three components remain, as the corners of the
  shrinking triangle touch the fold lines.

\item \(\beta = 3\) (elliptic umbilic).  The triangle at the centre
  has shrunk to a point, the minimum at zero in zero field has become
  a monkey saddle.

\item \(\beta > 3\).  The inner triangle reappears and grows again.
  For zero field there is a maximum at the uniform distribution, and
  three symmetric minimizers near the corners.
\end{itemize}

The series of transitions upon increasing inverse temperature fits to
the basic knowledge of the model without fields
\parencite{elliswang90}: We know that in zero field a) at very low
inverse temperature there is only one local minimum (and this is also
a global minimum) at the uniform distribution, b) at intermediate
inverse temperature there is a local minimum at the uniform
distribution and three symmetric minima near the corners c) at large
\(\beta\) there are only three symmetric minima near the corners.

The change from minimum to maximum of the uniform distribution under
increase of \(\beta\) is explained by an elliptic umbilic.
Additionally, for each of the minima at the corners there is an
additional fold line.  There must be a transition from the situation
of three non-intersecting rockets \(\beta \approx 2.2\) to an umbilic
plus three fold lines seen at the Ellis-Wang inverse
temperature~\(4\log 2\).  This is done with the help of the
three-symmetric-butterflies\,--\,beak-to-beak mechanism.

\subsection{Elements from singularity theory}

In order to derive and explain our results, concepts from singularity
theory will be useful.  The two most basic terms are \emph{catastrophe
  manifold} and\emph{ bifurcation set} of which the second term is
important since it is the basis for the metastable phase diagram.  But
first let us define the two: The \emph{catastrophe manifold} is the
set of \((\beta, \alpha, \nu)\) such that \(\nu\) is a stationary
point of \(f_{\beta, \alpha}\).  The \emph{bifurcation set} is the set
of \((\beta, \alpha)\) such that there exists a degenerate stationary
point \(\nu\) for \(f_{\beta, \alpha}\), that is, a stationary point
at which the Hessian has a zero eigenvalue.  The \emph{catastrophe
  map} \(\chi_{\beta}\) maps empirical spin distributions \(\nu\) to
a-priori measures \(\alpha\) such that the free energy
\(f_{\beta, \alpha}\) has a stationary point at \(\nu\).  We obtain
the expression for the catastrophe map by considering the zeros of the
differential of \(f_{\beta,\alpha}\).  For every tangent vector \(v\)
of the unit simplex we have
\begin{equation}
  \sum_{i=1}^3 v(\nu_i) (-\beta \nu_i - \log \alpha_i + \log \nu_i) = 0.
\end{equation}
We conclude that the second factor in the sum is a constant since the
\(v(\nu_i)\) sum up to zero.  Since \(\alpha\) is an element of the
unit simplex, we have \cref{eq:cata-map}.

The key idea from catastrophe theory is that at parameter values
belonging to the bifurcation set the stationary points of the function
change.  The most generic change is the fold where a minimum and a
maximum collide.  But the more parameters the potential has the more
complex behaviour is possible.  The famous theorem of Thom
\parencite[see Section~5 of Chapter~3 in][]{lu_intro_sing} lists these
possibilities for all potentials with at most four parameters.  We see
two of these so-called \emph{catastrophes} or \emph{singularities} in
the Potts model: the butterfly catastrophe and the elliptic umbilic.

\subsection{Constant-temperature slices of the bifurcation set}

The constant-temperature slices of the bifurcation set are
one-dimensional sets in the sense that we have a parametric
representation of the slices with one parameter but the curves show
pinches and self-intersections.  During the computation of these
curves we can already see some critical behaviour corresponding to the
beak-to-beak scenario and the elliptic umbilic point.

It is convenient to study the degenerate stationary points for fixed
temperature first.  From these points we can obtain the respective
slice of the bifurcation set via the catastrophe map \(\chi_{\beta}\).

\begin{theorem}
  \label{thm:const-temp-slices}
  The set of degenerate stationary points for any positive \(\beta\)
  is given by the symmetrized graph of the smooth function
  \(\gamma_{\beta}\) on \(D_\beta\) defined in
  Equations~(\ref{eq:gamma-curves}) and
  (\ref{eq:domain-gamma-curves}).
\end{theorem}

By the symmetrized graph of \(\gamma_{\beta}\) we mean the orbit
\(S_3\gamma_{\beta}(D_{\beta})\) under the permutation group \(S_3\).
Observe the critical behaviour that \(D_{\beta} = (0, 1)\) for
\(\beta = \frac{8}{3}\) (beak-to-beak) and that the two roots
\(1 - \frac{2}{\beta}\) and
\(\frac{1}{2} - \frac{1}{2}\sqrt{1 - \frac{8}{3\beta}}\) coincide for
\(\beta = 3\) (Elliptic umbilic).  We will now provide a series of
lemmata in preparation of the proof of this theorem.

The set of degenerate stationary points is determined by the so-called
\emph{degeneracy condition}. This condition states that the
determinant of the Hessian matrix at a stationary point vanishes.

\begin{lemma}
  \label{lem:deg-condition}
  The Hessian form of \(f_{\beta,\alpha}\) at the stationary point
  \(\nu\) is degenerate if and only if
  \begin{equation}
    \label{eq:deg-cond}
    3\nu_1\nu_2\nu_3\beta^2
    - 2(\nu_1 \nu_2 + \nu_2 \nu_3 + \nu_3 \nu_1) \beta + 1 = 0.
  \end{equation}
\end{lemma}

\begin{proof}
  Let \(\nu\) be a stationary point.  Choose \((\nu_1,\nu_2)\) as
  local coordinates for \(\nu\).  The Hessian form at \(\nu\) is
  represented with respect to the coordinate basis by the matrix
  \[\begin{pmatrix} \frac{1}{\nu_1} + \frac{1}{\nu_3} - 2\beta &
      \frac{1}{\nu_3} - \beta \\ \frac{1}{\nu_3} - \beta &
      \frac{1}{\nu_2} + \frac{1}{\nu_3} - 2\beta \end{pmatrix}.\]
  Calculating the determinant yields \[ 3\beta^2 -
    2\beta\left(\frac{1}{\nu_1} + \frac{1}{\nu_2} +
      \frac{1}{\nu_3}\right) + \frac{1}{\nu_1 \nu_2} + \frac{1}{\nu_1
      \nu_3} + \frac{1}{\nu_2 \nu_3}. \]
\end{proof}

If we rewrite the left-hand side of the degeneracy condition
(\ref{eq:deg-cond}) in \((\nu_1, \nu_2)\)-coordinates, it is a
quadratic function of \(\nu_2\) for fixed \(\beta\) and \(\nu_1\):
\begin{equation}
  \label{eq:deg-cond-coords}
  \beta(2-3\beta \nu_1) \nu_2^2 -
  \beta(2 - 3\beta\nu_1)(1-\nu_1)\nu_2 +
  1 - 2\beta\nu_1(1-\nu_1) = 0
\end{equation}
This equation has at most two solutions and one of them is given by
\begin{equation}
  \nu_2 = \gamma_{\beta}(\nu_1).
  \label{eq:quadratic-solution}
\end{equation}
The possible other solution is obtained by applying the respective
symmetry operation (exchanging the second and third component of
\(\nu\)).  The domain of \(\gamma_{\beta}\) is determined by the sign
of the discriminant and the additional condition that makes sure that
the result of \(\gamma_{\beta}\) is a point in the unit simplex.  Let
us investigate the latter condition first (Lemma~\ref{lem:formula})
and come to the condition imposed by nonnegativity of the discriminant
(Lemma~\ref{lem:discriminant}) afterwards.

First note that the solution formula~(\ref{eq:gamma-curves}) is not
defined for \(x = \frac{2}{3\beta}\) but converges to \(\pm\infty\)
because
\begin{equation}
  1 - 2 \beta x (1 - x) = \frac{1}{3}\left( \frac{8}{3\beta} - 1 \right)
  \begin{cases}
    > 0 & \beta < \frac{8}{3}\\
    = 0 & \beta = \frac{8}{3}\\
    < 0 & \beta > \frac{8}{3}
  \end{cases}
\end{equation}
except if \(\beta = \frac{8}{3}\) where the limit is
\(\frac{1}{4}\).  Furthermore, the domain of \(\gamma_{\beta}\)
must be such that \((x, \gamma_{\beta}(x), 1 - x -
\gamma_{\beta}(x))\) lies in the unit simplex, that is, we have to
analyse the following system of inequalities:
\begin{equation}
  \label{eq:unit-simplex-ineq}
  \begin{split}
    0 &< x < 1\\
    0 &< \gamma_{\beta}(x) < 1 - x
  \end{split}
\end{equation}

\begin{lemma}
  \label{lem:formula}
  Provided that \(\gamma_{\beta}(x)\) is a real number, we find:
  \begin{enumerate}
  \item For \(\beta < \frac{8}{3}\) the system (\ref{eq:unit-simplex-ineq}) is
    equivalent to
    \begin{equation}
      x \in
      \left( 0, \frac{2}{3\beta} \right) \cup
      \left( \frac{1}{2} - \frac{1}{2} \sqrt{1 - \frac{2}{\beta}},
        \frac{1}{2} + \frac{1}{2} \sqrt{1 - \frac{2}{\beta}}
      \right)
    \end{equation}
  \item For \(\frac{8}{3} \le \beta\) the system (\ref{eq:unit-simplex-ineq}) is
    equivalent to
    \begin{equation}
      x \in
      \left( 0,
        \frac{1}{2} - \frac{1}{2}\sqrt{1 - \frac{2}{\beta}}
      \right) \cup
      \left( \frac{2}{3\beta},
        \frac{1}{2} + \frac{1}{2} \sqrt{1 - \frac{2}{\beta}}
      \right)
    \end{equation}
  \end{enumerate}
\end{lemma}

\begin{proof}
  \begin{gather}
    0 < \gamma_{\beta}(x) < 1 - x\\
    \iff -(1-x) < \sqrt{(1 - x)^2 - \frac{4(1 - 2\beta x(1-x))}{\beta(2 - 3\beta x)}} < 1 - x
  \end{gather}

  The first inequality is trivially true because the square root is
  non-negative and \(1-x>0\).  Therefore we must only check the second
  inequality:

  \begin{gather}
    \sqrt{(1 - x)^2 - \frac{4(1 - 2\beta x(1-x))}{\beta(2 - 3\beta x)}} < 1 - x
    \\
    \iff
    0 < \frac{1 - 2\beta x(1-x)}{2 - 3\beta x}\\
  \end{gather}

  Suppose first \(2 - 3 \beta x > 0\). Then the inequality is
  equivalent to
  \begin{gather}
    1 - 2\beta x(1-x) > 0\\
    \iff 2\beta \left( x - \frac{1}{2} \right)^2 + 1 - \frac{\beta}{2} > 0
  \end{gather}
  As \(2\beta > 0\), this is an upfacing parabola whose minimal
  functional value is \(1 - \frac{\beta}{2}\) which is negative since
  \(\beta > 2\).  The roots of this parabola are:
  \begin{gather}
    2\beta(x - \frac{1}{2})^2 + 1 - \frac{\beta}{2} = 0\\
    \iff x = \frac{1}{2} \pm \frac{1}{2} \sqrt{1 - \frac{2}{\beta}}
  \end{gather}
  The solution is therefore the union of the two intervals \(0 < x <
  \frac{1}{2} - \frac{1}{2}\sqrt{1 - \frac{2}{\beta}}\) and
  \(\frac{1}{2} + \frac{1}{2}\sqrt{1 - \frac{2}{\beta}} < x <
  1\).

  Now, if \(2 - 3\beta x < 0\), we are looking for such \(x\) that the
  values of \(1 - 2\beta x (1-x)\) are negative.  This is the case if
  \begin{equation}
    \frac{1}{2} - \frac{1}{2}\sqrt{1 - \frac{2}{\beta}}
    < x < \frac{1}{2} + \frac{1}{2}\sqrt{1 - \frac{2}{\beta}}.
  \end{equation}

  To arrive at the claim in the lemma, we must investigate the order
  of \(\frac{2}{3\beta}\) and
  \(\frac{1}{2} \pm \frac{1}{2} \sqrt{1 - \frac{2}{\beta}}\).  We
  therefore analyse the inequality
  \begin{equation}
    \frac{1}{2} - \frac{1}{2} \sqrt{1 - \frac{2}{\beta}} > \frac{2}{3\beta}.
  \end{equation}
  Squaring both sides of the inequality reveals that it is equivalent
  to \(\beta < \frac{8}{3}\) which proves the claim.
\end{proof}

Let us continue with the analysis of the discriminant of the quadratic
equation~(\ref{eq:deg-cond-coords}). It is given by
\begin{equation}
  \begin{split}
    \beta^2(2-3\beta\nu_1)^2(1-\nu_1)^2 - 4\beta(2-3\beta\nu_1)(1-2\beta\nu_1(1-\nu_1))\\
    = \beta (3\beta \nu_1 - 2) (2 - \beta(1-\nu_1)) (2 - 3\beta \nu_1(1-\nu_1)).
  \end{split}
\end{equation}
We see that it has at most four possible roots depending on the
value of \(\beta\).  More precisely we find:

\begin{lemma}
  \label{lem:discriminant}
  Consider the function
  \begin{equation}
    \label{eq:eff_disc}
    x \mapsto (3\beta x - 2) (2 - \beta(1-x)) (2 - 3\beta x(1-x))
  \end{equation}
  for real \(x\).
  \begin{enumerate}
  \item For \(\beta < \frac{8}{3}\) this function has the two roots
    \(1 - \frac{2}{\beta} < \frac{2}{3\beta}\) and takes positive
    values only on
    \((-\infty, 1 - \frac{2}{\beta}) \cup (\frac{2}{3\beta},
    \infty)\).
  \item For \(\frac{8}{3} \le \beta < 3\) this function has the roots
    \(\frac{2}{3\beta} \le 1 - \frac{2}{\beta} < \frac{1}{2} -
    \frac{1}{2}\sqrt{1 - \frac{8}{3\beta}} \le \frac{1}{2} +
    \frac{1}{2}\sqrt{1 - \frac{8}{3\beta}}\) and takes positive values
    only on
    \[
      \left( -\infty, \frac{2}{3\beta} \right) \cup
      \left( 1 - \frac{2}{\beta},
        \frac{1}{2} - \frac{1}{2}\sqrt{1 - \frac{8}{3\beta}}
      \right) \cup
      \left( \frac{1}{2} + \frac{1}{2}\sqrt{1 - \frac{8}{3\beta}},
        \infty
      \right).
    \]  The equality of the roots is achieved for \(\beta = \frac{8}{3}\).
  \item For \(\beta \ge 3\) this function has the roots
    \(\frac{2}{3\beta} < \frac{1}{2} - \frac{1}{2}\sqrt{1 -
      \frac{8}{3\beta}} \le 1 - \frac{2}{\beta} < \frac{1}{2} +
    \frac{1}{2}\sqrt{1 - \frac{8}{3\beta}}\) and takes positive values
    only on
    \[
      \left( -\infty, \frac{2}{3\beta} \right) \cup
      \left( \frac{1}{2} - \frac{1}{2}\sqrt{1 - \frac{8}{3\beta}},
        1 - \frac{2}{\beta}
      \right) \cup
      \left( \frac{1}{2} + \frac{1}{2}\sqrt{1 - \frac{8}{3\beta}},
        \infty
      \right)
    \]
    The equality of the roots is achieved for \(\beta = 3\).
  \end{enumerate}
\end{lemma}

\begin{proof}
  The expression for the roots follow from the product form of the
  function.  Since we know all roots, the set where the function takes
  positive values is determined by the sign change at the roots.
  First, let us consider \(\beta < \frac{8}{3}\).  This implies the
  order of the roots since
  \(1 - \frac{2}{\beta} < \frac{1}{4} < \frac{2}{3\beta}\). The value
  of the derivative at the roots tells us how the sign changes.
  Denote the above function \eqref{eq:eff_disc} by \(g\).  Then we
  find:
  \begin{align*}
    g'\left( \frac{2}{3\beta} \right)    &= 4\left( \frac{8}{3} - \beta \right) \\
    g'\left( 1 - \frac{2}{\beta} \right) &= 12\left( \beta - \frac{8}{3} \right)(3 - \beta)
  \end{align*}
  This proves the case \(\beta < \frac{8}{3}\).

  Secondly, consider the case \(\frac{8}{3} \le \beta < 3\) but first
  assume \(\beta > \frac{8}{3}\).  The order of the of the roots
  \(\frac{2}{3\beta}\) and \(1-\frac{2}{\beta}\) now reverses and
  \(\beta < 3\) implies \(1 - \frac{2}{\beta} < \frac{1}{3} <
  \frac{1}{2} - \frac{1}{2}\sqrt{1 - \frac{8}{3\beta}}\).  Let
  us now analyse how the sign changes at the two largest roots.  Let
  \(x_0\) be any of the two roots of \(2 - 3\beta x(1-x)\).  We find
  \begin{equation}
    \label{eq:deriv_at_qu_roots}
    g'(x_0) = (3\beta x_0 - 2)(2 - \beta(1 - x_0)) \cdot 3\beta (2x_0 - 1).
  \end{equation}
  Since \(3\beta x - 2\) and \(2 - \beta(1 - x)\) are increasing and
  \(x_0\) is larger than their roots, the sign of \(g'(x_0)\) is
  determined by the sign of \(2x_0 - 1\) which is negative for
  \(x_0 = \frac{1}{2} - \frac{1}{2}\sqrt{1 - \frac{8}{3\beta}}\) and
  positive for
  \(x_0 = \frac{1}{2} + \frac{1}{2}\sqrt{1 - \frac{8}{3\beta}}\).  In
  the case \(\beta = \frac{8}{3}\) the roots \(\frac{2}{3\beta}\) and
  \(1 - \frac{2}{\beta}\) as well as the two largest roots coincide
  and for both the sign does not change.

  Let us now consider the last case \(3 \le \beta\) focusing on
  \(3 < \beta\) first.  Let us check the order of the roots.  Using
  the inequality \(\sqrt{1 + x} < 1 + \frac{1}{2}x\) for \(x < 0\), we
  find that
  \(\frac{2}{3\beta} < \frac{1}{2} - \frac{1}{2}\sqrt{1 -
    \frac{8}{3\beta}}\).  Similarly to the previous case the reversed
  inequality \(\beta > 3\) now implies the reversed inequality
  \(\frac{1}{2} - \frac{1}{2} \sqrt{1 - \frac{8}{3\beta}} <
  \frac{1}{3} < 1 - \frac{2}{\beta}\).  The last inequality for the
  roots is in fact equivalent to \(3 < \beta\):
  \begin{gather}
    1 - \frac{2}{\beta} < \frac{1}{2}\left(1 + \sqrt{1 - \frac{8}{3\beta}}
    \right)
    \\
    \iff
    1 + \frac{16}{\beta^2} - \frac{8}{\beta} < 1 - \frac{8}{3\beta}
    \\
    \iff
    3 < \beta
  \end{gather}

  Let us analyse how the sign changes at the two roots of
  \(2 - 3\beta x (1 - x)\).  This can be done using formula
  \eqref{eq:deriv_at_qu_roots}.  The first factor is positive for both
  roots since \(\frac{2}{3\beta}\) is the smallest root of the
  discriminant.  However, since \(1 - \frac{2}{\beta}\) lies in
  between the two roots of \(2 - 3\beta x (1 - x)\), the sign of
  \(g'\) is positive at both roots
  \(\frac{1}{2} \pm \frac{1}{2} \sqrt{1 - \frac{8}{3\beta}}\).  For
  \(\beta = 3\) the roots \(1 - \frac{2}{\beta}\) and
  \(\frac{1}{2} - \frac{1}{2} \sqrt{1 - \frac{8}{3\beta}}\) coincide
  and form a local maximum.
\end{proof}

\begin{proof}[Proof of \cref{thm:const-temp-slices}]
  The expression for \(\gamma_{\beta}(x)\)
  (formula~(\ref{eq:gamma-curves})) is obtained by solving the
  quadratic equation~(\ref{eq:deg-cond-coords}) for \(\nu_2\).  In
  order to determine the domain \(D_{\beta}\) we have to intersect the
  sets from Lemma \ref{lem:formula}, which ensures that
  \((x, \gamma_\beta(x), 1 - x - \gamma_\beta(x))\) lies in the unit
  simplex, and Lemma \ref{lem:discriminant} ensuring the nonnegativity
  of the discriminant.
\end{proof}

\subsection{Computation of the critical temperatures}

We discuss the various critical temperatures in increasing order.

\subsubsection{The butterfly temperature}
\label{sec:butterfly-taylor}

Looking at the constant-temperature slices of the bifurcation set in
the regime \(2 < \beta <\frac{8}{3}\), we find a qualitative change of
the curve (compare \cref{fig:bifu-slice}): A pentagram-like shape
unfolds.  The butterfly temperature is defined by the \(\beta\) at
which this happens which is for \(\beta = \frac{18}{7}\).  This can be
seen by a Taylor expansion of the curve which describes the
constant-temperature slices as the coefficients undergo sign changes.
Because of symmetry, it does not matter which of the three rockets we
consider.  Let us consider the degenerate stationary points with
\(\nu_3 \le \min\{\nu_1,\nu_2\}\).  More precisely, consider the
degeneracy equation \eqref{eq:deg-cond} in the following coordinates:
\begin{equation}
  \label{eq:xy-coords}
  \begin{split}
    x &= \frac{\sqrt{3}}{2}(\nu_1-\nu_2)\\
    y &= \frac{1}{2}(3\nu_3-1)
  \end{split}
\end{equation}
In these coordinates the unit simplex \(\Delta^2\) is an equilateral
triangle with center at the origin.  The equation then reads
\begin{equation}
  \label{eq:deg_xy_coords}
  \frac{1}{9} \left(6 \beta \left(x^2+y^2-1\right)
    +\beta ^2 (2 y+1) \left((y-1)^2-3 x^2\right)+9
  \right)=0.
\end{equation}
For \(x = 0\), this equation reads
\begin{equation}
  2\beta^2 y^3 + 3\beta(2 - \beta)y^2 + (\beta - 3)^2 = 0.
\end{equation}
which has a single, negative root \(y_0\).  Using the Implicit
Function Theorem the set of degenerate stationary points locally
around \((0, y_0)\) is the graph of a function \(y = g_{\beta}(x)\)
which solves \eqref{eq:deg_xy_coords}.  It is of course also possible
to obtain the values of \(g_\beta^{(n)}(0)\) using Implicit
Differentiation which allows us to write down a Taylor expansion for
\(g_{\beta}\).  If we plug this into the catastrophe map
\(\chi_{\beta}\) we arrive at an expansion for the respective slice of
the bifurcation set:
\begin{equation}
  \begin{split}
    \chi_{\beta}&(x, g_{\beta}(x)) - \chi_{\beta}(0, y_0) = \\
    -\begin{pmatrix}
      1\\ 1\end{pmatrix} &\left(\frac{27}{2} \left(\beta -\frac{18}{7}\right)+O\left(\left(\beta
          -\frac{18}{7}\right)^2\right)\right) x^2 +\\
    &\begin{pmatrix}-1\\1\end{pmatrix} \left(\frac{36}{7} \sqrt{3} \left(\beta -\frac{18}{7}\right)+O\left(\left(\beta -\frac{18}{7}\right)^2\right)\right) x^3 +\\
    &\begin{pmatrix}1\\1\end{pmatrix} \left(\frac{39366}{2401}-\frac{48519}{343} \left(\beta -\frac{18}{7}\right)+O\left(\left(\beta -\frac{18}{7}\right)^2\right)\right)x^4 + O\left(x^5\right)
  \end{split}
\end{equation}
This has been achieved using exact computations in Mathematica.  Here
we have slightly abused notation by writing \(\chi_{\beta}\) for the
coordinate representation of \(\chi_{\beta}\) using
\((x, y)\)-coordinates in the source and the following coordinates in
the target space:
\begin{equation}
  \label{eq:uv-coords}
  \begin{split}
    u &= \log\frac{\alpha_1}{\alpha_3}\\
    v &= \log\frac{\alpha_2}{\alpha_3}
  \end{split}
\end{equation}
So the coordinate representation of the catastrophe
map~\(\chi_{\beta}\) is given by
\begin{equation}
  \begin{split}
    (x, y) \mapsto \begin{pmatrix}
      \log\frac{1 - y - \sqrt{3}x}{1 + 2y} + \beta\left( y + \frac{x}{\sqrt{3}} \right)
      \\
       \log\frac{1 - y + \sqrt{3}x}{1 + 2y} + \beta\left( y - \frac{x}{\sqrt{3}} \right)
    \end{pmatrix}.
  \end{split}
\end{equation}

\subsubsection{The crossing temperature}
\label{sec:crossing-temp}

\begin{lemma}
  The inverse crossing temperature is given by
  \begin{equation}
    \beta_\mathrm{cross} = \frac{3}{(1+2s)(1-s)} \approx 2.74564
  \end{equation}
  where \(s\) is the unique root in \((0,1)\) of
  \begin{equation}
    \frac{3x}{(1+2x)(1-x)} - \log\frac{1+2x}{1-x}.
  \end{equation}
\end{lemma}

\begin{proof}
  Let \(\gamma(s) = \frac{1}{3}(1+2s,1-s,1-s)\) for \(0 \le s < 1\)
  and let \(\alpha\) be the uniform distribution.  The inverse
  crossing temperature \(\beta_{\mathrm{cross}}\) equals the \(\beta\)
  such that the two outer local extrema of
  \(f_{\beta, \alpha}\circ\gamma\) annihilate.  This is characterized
  by the two equations (first and second derivatives of
  \(f_{\beta, \alpha}\circ\gamma\))
  \begin{align}
    \label{eq:fold-1st-deriv}
    \beta s - \log\frac{1+2s}{1-s} &= 0 \\
    \label{eq:fold-2nd-deriv}
    \beta - \frac{3}{(1+2s)(1-s)} &= 0.
  \end{align}

  Plugging \eqref{eq:fold-2nd-deriv} into \eqref{eq:fold-1st-deriv}
  motivates the following definition: Let \(g\) be given on \((0,1)\)
  by
  \begin{equation}
    g(x) = \frac{3x}{(1+2x)(1-x)} - \log\frac{1+2x}{1-x}.
  \end{equation}
  The first derivative of this function vanishes in the open interval
  \((0, 1)\) exactly at \(x = \frac{1}{4}\), it is decreasing on
  \(\left(0, \frac{1}{4}\right)\), increasing on
  \(\left(\frac{1}{4}, 1\right)\) and \(\lim_{x\to 0} g(x) = 0\).
  Therefore \(g\) has a unique root in \((0,1)\).
\end{proof}

\subsubsection{The triangle-touch temperature}
\label{sec:triangle-touch-temp}

The \emph{triangle-touch temperature} is defined as the temperature
\(1/\beta_{\mathrm{touch}}\) such that in the respective
constant-temperature slice the vertices of the central triangle touch
the fold lines.  By definition,
\(\frac{8}{3} < \beta_{\mathrm{touch}} < 3\).

\begin{lemma}
  The inverse triangle-touch temperature \(\beta_{\mathrm{touch}}\) is
  the unique zero in \([\frac{8}{3}, 3]\) of
  \begin{equation}
    2\operatorname{artanh}\sqrt{1 - \frac{8}{3x}}
    + \frac{3x}{4}\left( 1 - \sqrt{1 - \frac{8}{3x}} \right)
    - \log\left(\frac{x - 2}{2}\right) - 3.
  \end{equation}
\end{lemma}

\begin{proof}
  First, observe that the function is strictly increasing, positive
  for \(x = \frac{8}{3}\) and negative for \(x = 3\).  The function
  values are \(1 - \log 3\) and \(\frac{3}{2} - \log 4\) respectively.
  Therefore this function has a unique zero in the specified interval.

  It suffices to show that one vertex of the central triangle and one
  of the fold lines meet because of symmetry.  Since the vertex lies
  on an axis of symmetry of the simplex for all
  \(\frac{8}{3} < \beta < 3\), we know that the intersection point
  with the fold line must also lie on the axis of symmetry (the centre
  of the fold line).  For the space of a-priori measures \(\alpha\) we
  use the following coordinates:
  \begin{equation}
    \label{eq:sym-fields-coords}
    \begin{split}
      p &= \sqrt{3} \log \frac{\alpha_1}{\alpha_2}\\
      q &= \log \frac{\alpha_1 \alpha_2}{\alpha_3^2}
    \end{split}
  \end{equation}
  The vertex of the triangle fulfills
  \(p = \sqrt{3}(\log(\beta - 2) + 3 - \beta)\) and the centre of the
  fold line has
  \(p = \sqrt{3}\left( \log 2 - \frac{1}{4}\beta - \frac{3\beta}{4}
    \sqrt{1 - \frac{8}{3\beta}} + 2\operatorname{artanh}\sqrt{1 -
      \frac{8}{3\beta}} \right)\).  Equating the two formulas proves
  the claim.  The values of the \(u\)-coordinate can be calculated
  from the respective degenerate stationary points (see
  \cref{fig:triangle_touch_explanation}).  The \(\nu_1\) values are
  the lower bounds of the domains and
  \(\nu_2 = \frac{1}{2}(1 - \nu_1)\).
\end{proof}

\begin{figure}
  \centering
  \includegraphics{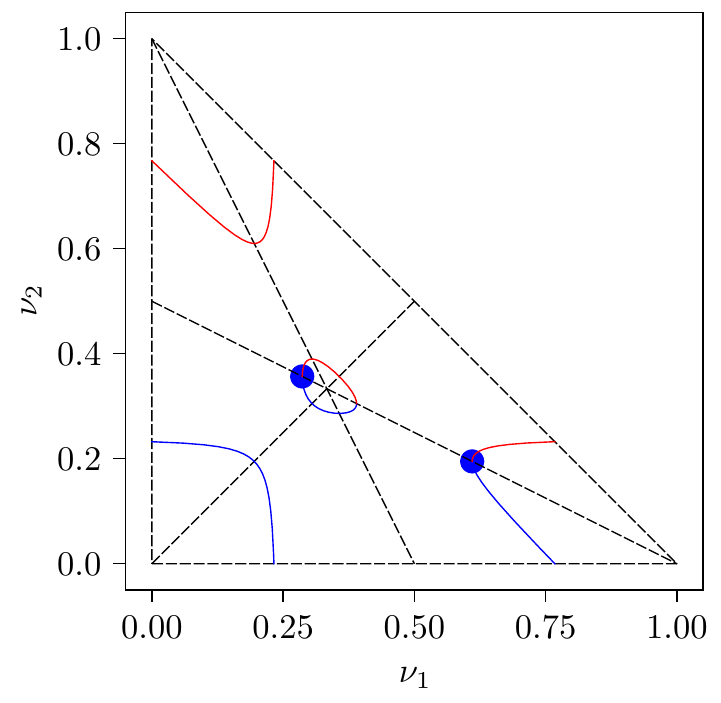}
  \caption{The figure shows the two degenerate stationary points in
    \(\nu\)-space that are mapped to a touching point of the inner
    triangle with of a fold line in \(\alpha\)-space under the
    catastrophe map \(\chi_\beta\).}
  \label{fig:triangle_touch_explanation}
\end{figure}

\subsubsection{The elliptic umbilic temperature}

We know from singularity theory that the elliptic umbilic is a doubly
degenerate point, that is, a point where the Hessian has two zero
eigenvalues.  There is only one such point for this potential and it
is given by \(\beta = 3\): The vanishing mixed second-order partial
derivatives of the potential implies \(\beta = \frac{1}{\nu_3}\).
Plugging this into
\[\frac{\partial^2 f_{\beta,\alpha}}{\partial \nu_i^2} = 2\beta -
  \frac{1}{\nu_i} - \frac{1}{\nu_3} = 0\] yields
\(\frac{1}{\nu_i} = \beta\).  These partial derivatives vanish
therefore only for \(\beta = 3\) and \(\nu = \frac{1}{3}(1,1,1)\) in
zero magnetic field.  Furthermore, the third-order Taylor expansion of
the potential \(f_{\beta, \alpha}\) for \(\beta = 3\) and zero
magnetic field in the \((x, y)\)-coordinates (\ref{eq:xy-coords}) is
given by \(-\frac{1}{3}y^3 + x^2y + \mathrm{const}\).  This is exactly
the germ of the elliptic umbilic from Thom's seven elementary
catastrophes.

\subsection{A parametric representation of the bifurcation set}

As we have learned in the previous subsections, the extended phase
diagram is constructed via the bifurcation set.  In this subsection we
present the parametric representation which was used to create
Figure~\ref{fig:3d-bifu-set} and how it is obtained.

\begin{theorem}
  The bifurcation set is given by the union
  \(F^+(\Delta^2)\cup F^-(\Delta^2)\) of the images of the two maps
  \(F^{\pm}\) from \(\Delta^2\) to the parameter space
  \((0,\infty) \times \Delta^2\) with components
  \begin{align}
    F^{\pm}_1(\nu) &= \frac{1}{3}\sum_{i=1}^3 \frac{1}{\nu_i} \pm \sqrt{\frac{1}{9}\left(\sum_{i=1}^3 \frac{1}{\nu_i}\right)^2 - \frac{1}{3} \sum_{i<j} \frac{1}{\nu_i \nu_j}} \\
    F^{\pm}_2(\nu) &= \chi_{F_1^{\pm}(\nu)}(\nu) = \left( \frac{\nu_i e^{-F^{\pm}_1(\nu) \nu_i}}{\sum_{k=1}^q \nu_k e^{-F^{\pm}_1(\nu) \nu_k}} \right)_{i=1}^3
  \end{align}
\end{theorem}

\begin{proof}
  First, let us check that \(F^+\) and \(F^-\) map \(\Delta^2\) into
  \((0, \infty) \times \Delta^2\).  Clearly, \(F^{\pm}_2(\nu)\) is
  an element of \(\Delta^2\).  Furthermore, \(F^+_1(\nu)\) is
  obviously positive.  Since every \(\nu_i\) is positive, we have
  \begin{equation}
    \sqrt{
      \frac{1}{9}\left(\sum_{i=1}^3 \frac{1}{\nu_i}\right)^2
      - \frac{1}{3} \sum_{i<j} \frac{1}{\nu_i \nu_j}
    } < \frac{1}{3}\sum_{i=1}^3 \frac{1}{\nu_i}
  \end{equation}
  which implies that \(F^-_1(\nu)\) is also positive.  A point
  \((\beta, \alpha)\) belongs to the bifurcation set if and only there
  exists a degenerate stationary point \(\nu\) in the unit simplex
  such that \(\alpha = \chi_{\beta}(\nu)\).  \Cref{lem:deg-condition}
  shows us that the degeneracy condition~\eqref{eq:deg-cond} is a
  quadratic equation but this time considered as a function of
  \(\beta\).  Note also that it is independent of \(\alpha\).  The
  discriminant is given by
  \begin{align}
    4 (\nu_1 \nu_2 +\nu_2 \nu_3 +\nu_3 \nu_1)^2 - 12 \nu_1 \nu_2 \nu_3
    &= 4((\nu_1 \nu_2 +\nu_2 \nu_3 +\nu_3 \nu_1)^2 - 3 \nu_1 \nu_2 \nu_3)
    \\
    &= 4((\nu_1 \nu_2)^2 +(\nu_2 \nu_3)^2 - \nu_1 \nu_2 \nu_3 + (\nu_3 \nu_1)^2).
  \end{align}
  On the boundary of the simplex, the discriminant is positive except
  at the vertices.  We see by calculus that it achieves its minimal
  value at the centre of the simplex where it takes the value zero.
  Therefore the quadratic equation for \(\beta\) has the two solutions
  \(F^\pm_1(\nu)\) for every \(\nu\).

  Now that we know the two possible \(\beta\)-values that make \(\nu\)
  fulfill the degeneracy condition we can use the catastrophe map
  \(\chi_{\beta}\) to obtain the respective a-priori measures
  \(\alpha\) to make them degenerate stationary points.  This proves
  the claim.
\end{proof}

\section{The stable phase diagram}
\label{sec:classical-diagram}

The classical phase diagram is a partition of the parameter space.
However, in contrast to the metastable phase diagram, the cells of this
partition contain \((\beta, \alpha)\) such that the number of
\emph{global} minimizers stays constant inside the cell.  This type of
phase diagram does not see the fine bifurcation behaviour of the rate
function and is therefore much simpler to describe.  We can think of
the classical phase diagram as given by a “surface” in the parameter
space: the coexistence surface.  On this surface we have a coexistence
of at least two phases.  It is therefore clear that the coexistence
surface lies in those cells of the metastable phase diagram in which the
rate function has at least two minimizers.  The complement of the
coexistence surface defines a region of the parameter space in which
the rate function has a unique global minimum.

The surface is best understood by moving in the direction of
increasing \(\beta\) (see
\cref{fig:maxwell-rockets,fig:maxwell-pentagram,fig:maxwell-beak-ew,fig:maxwell-beyond}).
Leaving the high temperature regime (\(\beta \le 2\)), it consists of
three lines on the axis of symmetry.  These lines have progeny, namely
two lines emerging at a positive angle.  This results in three
Y-shaped sets composed of curved and straight lines.  Furthermore,
each of the offspring lines of the Y-shaped curves meet during the
beak-to-beak scenario we have already seen in the extended phase
diagram and form a triangle.  Finally, the triangle shrinks to a point
and we see a star-shaped set consisting of three straight lines (the
axes of symmetry).  On the coexistence surface the rate function has
at least two and at most four global minimizers.  The point of
coexistence of four phases is the well-known Ellis-Wang point
\parencite{elliswang90}.  For \(2 < \beta \le \frac{18}{7}\) it is
only possible to have a coexistence of two phases.  However, starting
with \(\beta\) larger than \(\frac{18}{7}\) we find so-called triple
points (coexistence of three phases).  These will be important for our
numerical computation of the Maxwell set (coexistence surface).  Let
us summarize our result:

\begin{theorem}
  Let us use the \((x, y)\)-coordinates~\eqref{eq:xy-coords} for the
  \(\alpha\)-simplex.
  \begin{enumerate}
  \item For \(\beta \le 2\) the rate function has a unique global
    minimum for any \(\alpha\).
  \item For \(2 < \beta \le \frac{18}{7}\) the rate function has
    precisely two global minimizers if \(\alpha\) lies in the segment
    \(\{0\} \times \left(-\frac{1}{2}, -\frac{1 - (\beta -
        2)e^{3-\beta}}{2 + (\beta - 2)e^{3-\beta}} \right)\) or its
    images under the permutation group.  For any other \(\alpha\) the
    rate function has a unique global minimum (see
    \cref{fig:maxwell-rockets}).
  \item For \(\frac{18}{7} < \beta < 4\log 2\) the rate function has
    precisely two global minimizers if \(\alpha\) lies in the segment
    \( \{0\} \times (-\frac{1}{2}, y_{\mathrm{triple}}(\beta))\) or if
    \(\alpha\) lies on the curve that is a solution to the initial
    value problem \eqref{eq:maxwell-ivp} or its images on the
    permutation group.  Here, \(y_{\mathrm{triple}}(\beta)\) is the
    \(y\)-coordinate of the triple point.  If \(\alpha\) is a
    triple point, the rate function has three global minimizers.
    For any other \(\alpha\) it has a unique global minimum. (see
    \cref{fig:maxwell-pentagram,fig:maxwell-beak-ew}).
  \item For \(\beta \ge 4\log 2\) the rate function has two global
    minimizers for any \(\alpha\) on the segment
    \((-\frac{1}{2}, 0)\times \{0\}\) or its images under the
    permutation group.  It has four global minimizers if
    \(\beta = 4\log 2\) and three global minimizers if
    \(\beta > 4\log 2\) for
    \(\alpha = (\frac{1}{3}, \frac{1}{3}, \frac{1}{3})\).  For any
    other \(\alpha\) it has a unique global minimum (see
    \cref{fig:maxwell-beyond}).
  \end{enumerate}
\end{theorem}

Now it is clear from the previous section that for \(\beta \le 2\) we
do not see multiple global minimizers because points in the
bifurcation set have \(\beta > 2\) and in a high temperature regime we
have a unique global minimum.  The other three regimes are interesting
and it is useful to keep the bifurcation set in mind when analysing
these.  However, before we discuss the other regimes in detail let us
state another observation which clarifies the word \enquote{surface}
of the term \emph{coexistence surface:} Locally (except at the triple
points), this set is indeed a two-dimensional submanifold.

Suppose \((\beta, \alpha)\) is such that the rate function
\(f_{\beta, \alpha}\) has two distinct non-degenerate stationary
points \(\nu_{\beta, \alpha}\) and \(\mu_{\beta, \alpha}\).  Since the
rate function depends smoothly on its parameters, the Implicit
Function Theorem tells us that we find two smooth maps
\((\beta, \alpha) \mapsto \nu_{\beta, \alpha}\) and
\((\beta, \alpha) \mapsto \mu_{\beta, \alpha}\) that map a
neighbourhood \(U\) of \((\beta, \alpha)\) to \(\Delta^2\) such that
\(\nu_{\beta, \alpha}\) and \(\mu_{\beta, \alpha}\) are stationary
points of \(f_{\beta, \alpha}\) for every \((\beta, \alpha)\) in the
neighbourhood.

\begin{lemma}
  The set of \((\beta,\alpha)\) such that
  \(f_{\beta, \alpha}(\nu_{\beta,\alpha}) = f_{\beta,
    \alpha}(\mu_{\beta,\alpha})\) is a two-dimensional embedded
  submanifold of \((0,\infty) \times \Delta^2\).
\end{lemma}

\begin{proof}
  Let us define the smooth map \(F\) from \(U\) to \(\mathbb{R}\) via
  \begin{equation}
    F(\beta, \alpha) =
      f_{\beta, \alpha}(\nu_{\beta,\alpha}) - f_{\beta, \alpha}(\mu_{\beta,\alpha}).
  \end{equation}
  We now want to apply the \emph{Constant-Rank Level Set Theorem}
  \parencite{Lee_2012} for \(F\) to conclude the proof.  The
  differential of \(F\) in terms of the \((\beta, u, v)\) coordinates
  is given via the row vector
  \begin{equation}
    \begin{pmatrix}
      -\frac{1}{2}(\|\nu_{\beta,\alpha}\|^2 - \|\mu_{\beta,\alpha}\|^2), &
      (\mu_{\beta,\alpha})_1 - (\nu_{\beta,\alpha})_1, &
      (\mu_{\beta,\alpha})_2 - (\nu_{\beta,\alpha})_2
    \end{pmatrix}
  \end{equation}
  which is the zero map if and only if
  \(\nu_{\beta, \alpha} = \mu_{\beta,\alpha}\).  Thus, the
  differential has constant rank one.
\end{proof}

\subsection{Coexistence in the regime of the rockets}

In the regime of the rockets (\(2 < \beta \le \frac{18}{7}\)) the only
cell that yields two local minimizers is given by the region enclosed
by the rockets.  The Maxwell set of this region is given by the
intersection with the axes of symmetry.  This is due to the fact that
the two local minimizers lie in different fundamental cells of the
simplex and that an asymmetry in the fields \(\alpha\) leads to the
same asymmetry in the global minimizer.  This is explained by the
following lemma which is inspired by Lemma~1 of \cite{wang94}.

\begin{lemma}[Tilting Lemma]
  \label{lem:tilting}
  Let \(\nu\) be a global minimum of \(f_{\beta,\alpha}\).  If
  \(\alpha_i > \alpha_j\), then \(\nu_i > \nu_j\).
\end{lemma}

\begin{proof}
  Let \(\nu'_j = \nu_i\), \(\nu'_i = \nu_j\) and \(\nu'_k = \nu_k\)
  for \(k\) not in \(\{i,j\}\).  Since
  \begin{equation}
    f_{\beta,\alpha}(\nu') - f_{\beta,\alpha}(\nu) =
      (\nu_i - \nu_j)(\log\alpha_i - \log\alpha_j) \ge 0
  \end{equation}
  and \(\log\alpha_i - \log\alpha_j > 0\), we conclude
  \(\nu_i \ge \nu_j\).  Assume \(\nu_i = \nu_j\) and consider the
  push-forward of the tangent vector
  \(v = \frac{\partial}{\partial\nu_j} -
  \frac{\partial}{\partial\nu_i}\):
  \begin{equation}
    df_{\beta,\alpha}(v) = \log\alpha_i - \log\alpha_j > 0.
  \end{equation}
  Thus \(\nu\) is not a stationary point which contradicts the fact
  that \(\nu\) is a minimizer.  Therefore \(\nu_i > \nu_j\).
\end{proof}

The cusp point of the rockets is given by the end point of the curve
\(\chi_{\beta} \circ \gamma_{\beta}(\nu_1)\) which in this regime is
\(\nu_1 = 1 - \frac{2}{\beta}\) (see
Theorem~\ref{thm:const-temp-slices}).  Thus the Maxwell set in
\((x, y)\)-coordinates is the segment
\begin{equation}
  \{0\} \times \left(-\frac{1}{2},
    -\frac{1 - (\beta - 2)e^{3-\beta}}{2 + (\beta - 2)e^{3-\beta}} \right).
\end{equation}

\begin{figure}
  \centering
  \includegraphics{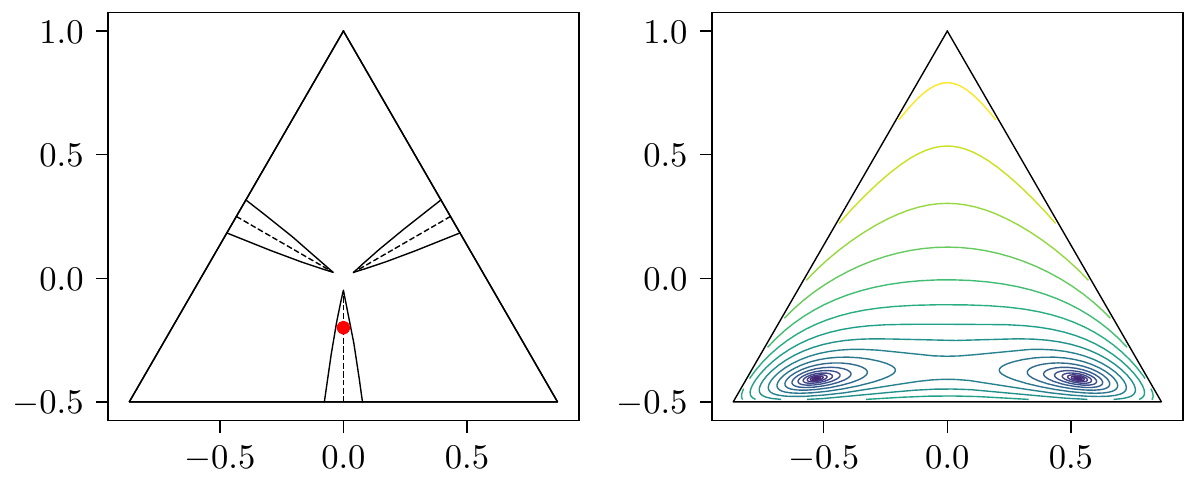}
  \caption{The left plot shows the Maxwell set (dashed lines) inside
    of the rockets.  The right plot shows the contours of a typical
    point on the Maxwell set as indicated by the red dot in the left
    plot.  The minimizers are equally deep and are mapped onto each
    other under reflection at the vertical axis.}
  \label{fig:maxwell-rockets}
\end{figure}

\subsection{Coexistence in the regime of disconnected pentagrams}
\label{sec:org9dafc7f}

In the regime \(\frac{18}{7} < \beta < \frac{8}{3}\) three pentagrams
have already unfolded but are still disconnected.  As we have already
discussed there are (modulo symmetry) three cells with two local
minimizers and one cell with three local minimizers.  The Maxwell set
in cell I (see Figure~\ref{fig:maxwell-pentagram}) is the easiest.
Here we have again two minimizers in different fundamental cells and
therefore the Maxwell set is the intersection with the respective axis
of symmetry.  In the cells III and IV we also have two local
minimizers but they lie in the same fundamental cell.  Therefore the
Tilting Lemma (Lemma~\ref{lem:tilting}) does not apply and the Maxwell
set is a curved line deviating from the axis of symmetry.  Cell II is
special because we have three different local minimizers two of which
lie in different fundamental cells.  Since the classical phase diagram
describes the degeneracy of the \emph{global} minimum, we know that
the Maxwell set continues on the axis of symmetry for as long as the
two local minima from the different fundamental cells are lower than
the third minimum.  There exists however a point on the axis of
symmetry at which this behaviour changes: the triple point.  This is a
point at which all minimizers are global minimizers.  Since two of the
local minimizers lie in different fundamental cells, the triple point
must lie on the “star” (Lemma \ref{lem:tilting}), that is, at least
two components are equal.  This is also the point where the Maxwell
set leaves the axis of symmetry because the minimizers involved do not
lie in different fundamental cells anymore.  It suffices to compute
the Maxwell set in either cell III or IV because of symmetry.  The
problem of computing the Maxwell set can be transformed into a
solution of an initial value problem where the initial value is given
by the triple point.

\begin{figure}
  \centering
  \includegraphics{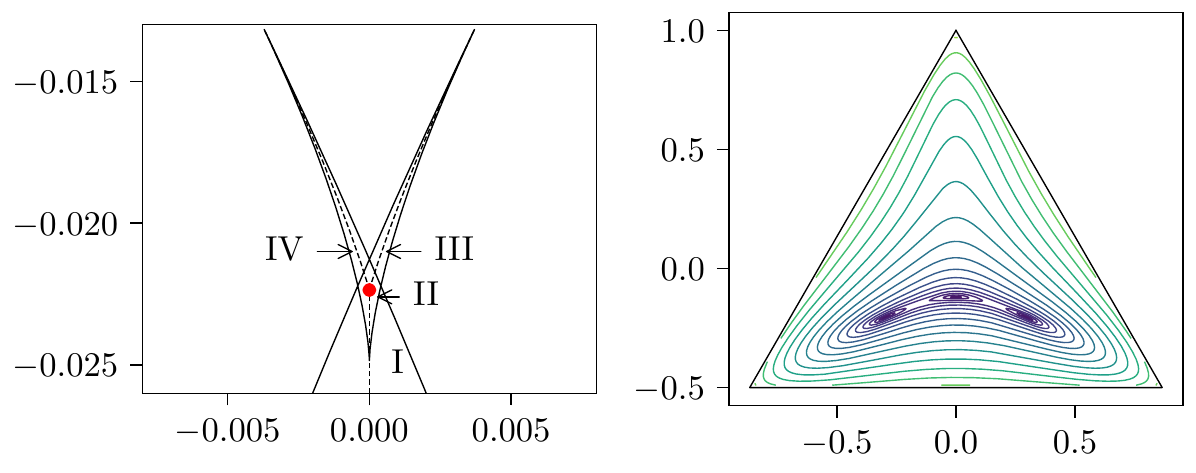}
  \caption{The left plot shows a magnification of one of the three
    pentagrams with its containing Maxwell set (dashed lines).  The
    red dot marks the triple point at which all three minimizers are
    global minimizers.  This can be seen in the right plot which shows
    a contour plot of the potential at the triple point.}
  \label{fig:maxwell-pentagram}
\end{figure}

\begin{proposition}
  For each positive \(\beta\) there exists exactly one point
  \(\alpha\) with \(\alpha_1 \le \alpha_2 \le \alpha_3\) such that
  \(f_{\beta, \alpha}\) has precisely three global minimizers.
\end{proposition}

\begin{proof}
  It is clear that the triple point \(\alpha\) lies on the axis of
  symmetry.  Therefore, let the curve \(u \mapsto \alpha(u)\) be the
  axis of symmetry intersected with cell II.  Then we have two local
  minimizers \(\mu(u)\) and \(\nu(u)\) of \(f_{\beta, \alpha(u)}\)
  such that \(\mu_1(u) > \nu_1(u)\).  The difference
  \begin{equation}
    g(u) := f_{\beta, \alpha(u)}(\nu(u)) - f_{\beta, \alpha(u)}(\mu(u))
  \end{equation}
  is a monotonically increasing function of \(u\):
  \begin{equation}
    g'(u) = \frac{\partial f_{\beta,\alpha(u)}}{\partial u} = \mu_1(u) - \nu_1(u) > 0.
  \end{equation}
\end{proof}

\begin{proposition}
  The set of all \(\alpha\) such that the two local minimizers inside
  of the fundamental cell \(\nu_1 < \nu_2 < \nu_3\) have the same
  depth is given in \((u, v)\)-coordinates by the graph of
  \(u \mapsto v(u)\) which is the solution of the initial value
  problem
  \begin{equation}
    \label{eq:maxwell-ivp}
    \begin{split}
      \frac{dv}{du} &=
        -\frac{\nu_1(u, v) - \mu_1(u, v)}{\nu_2(u, v) - \mu_2(u, v)}
      \\
      v(u_0) &= v_0
    \end{split}
  \end{equation}
  where \(\alpha = (u_0, v_0)\) is the triple point, \(\nu(u, v)\) and
  \(\mu(u, v)\) are the two local minimizers of \(f_{\beta, \alpha}\)
  in the same cell for \(\alpha = (u, v(u))\).
\end{proposition}

\begin{proof}
  The curve fulfills
  \begin{equation}
    f_{\beta,\alpha(u)}(\nu(u, v)) = f_{\beta,\alpha(u)}(\mu(u, v)).
  \end{equation}
  Note that
  \begin{equation}
    \frac{d}{du} f_{\beta,\alpha(u)}(\nu(u, v)) = \nu_1(u, v) - \nu_2(u, v) \frac{dv}{du}
  \end{equation}
  since \(\nu(u, v)\) is a stationary point.  This proves the
  proposition.
\end{proof}

However, numerically computing the Maxwell set using this
characterization is difficult because the mapping
\(u \mapsto (\nu(u, v), \mu(u, v))\) is not explicit.  Therefore we
consider the system of equations
\begin{align*}
  f_{\beta,\alpha}(\mu) &= f_{\beta,\alpha}(\nu)\\
  \chi_{\beta}(\mu) &= \chi_{\beta}(\nu)
\end{align*}
instead.  Because the free energy at a stationary point is given by
\begin{equation}
  f_{\beta,\alpha}(\nu) = \sum_{i=1}^3 \left(
    \frac{1}{2}\beta \nu_i^2 + \nu_i \log \sum\limits_{j=1}^3 \nu_j e^{-\beta\nu_j}
  \right)
\end{equation}
this system does not depend on \(\alpha\) and can be solved
numerically for fixed values of \(\beta\) and \(\mu_1\).  For better
stability of the numerics we start with the triple point where the
minimizers are well separated and then iteratively use the results as
initial values for the next numerical step.  In this way, the
\cref{fig:maxwell-rockets,fig:maxwell-pentagram,fig:maxwell-beak-ew,fig:maxwell-beyond}
were obtained.

\subsection{From beak-to-beak to Ellis-Wang}

The qualitative nature of the Maxwell sets does not change with
increasing \(\beta\) after the beak-to-beak scenario until we reach
the Ellis-Wang point.  The cells with two minima which resulted from
the merging of the horns of the pentagrams contain two minimizers
which lie in the same fundamental cell as discussed in the previous
subsection.  Therefore the Maxwell set is again given by the solution
to the initial value problem~\eqref{eq:maxwell-ivp}.  This continues
even after the crossing temperature where the central cell now
contains four minima.  Before the Ellis-Wang temperature the central
fourth minimum is a local but not global minimum.  The two outer
minima each lie in the same fundamental cell so that the initial value
problem applies.  However, this changes after the Ellis-Wang point.

\begin{figure}[t!]
  \centering
  \includegraphics{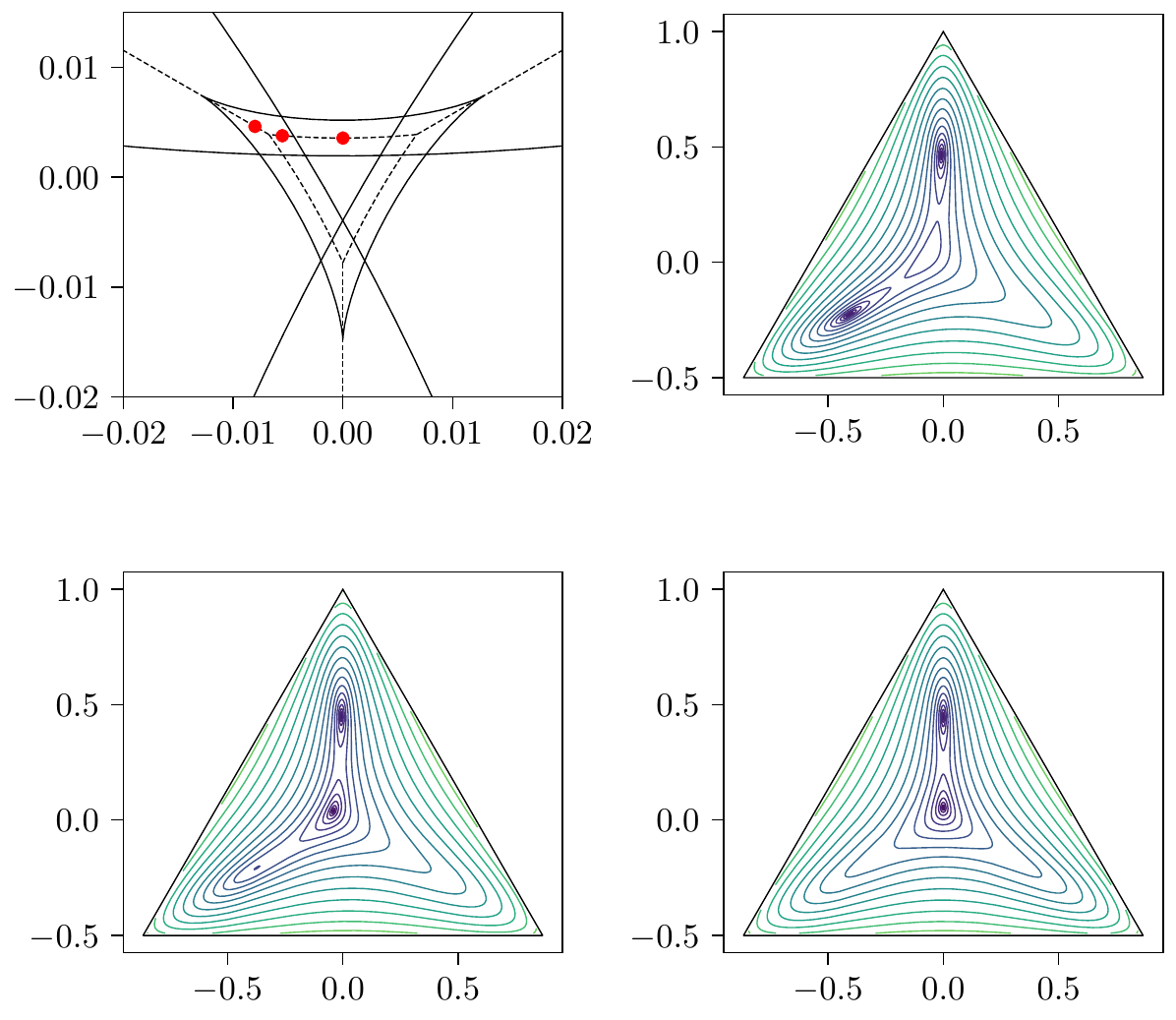}
  \caption{The upper left plot shows a magnification of the center of
    the simplex showing the bifurcation set together with the Maxwell
    set (dashed lines) for a \(\beta\) between beak-to-beak and
    Ellis-Wang.  Following the red dots on the Maxwell set from left
    to right we see how the potential changes in the upper right,
    lower left and lower right plots.}
  \label{fig:maxwell-beak-ew}
\end{figure}

\FloatBarrier

\subsection{Beyond Ellis-Wang}

At \(\beta = 4 \log2\) the outer minima and the central local minimum
are equally deep.  After the Ellis-Wang point (\(\beta > 4 \log 2\)),
the outer minima are lower than the central minimum.  They are equally
deep in zero magnetic field and it is still possible using the Tilting
Lemma to break the symmetry of the fields partially and achieve two
equally deep minimizers.  This can be done in the regime where the
rate function has a local minimum in the centre (\(\beta < 3\)) as
well as in the regime where the local minimum has become a local
maximum (\(\beta > 3\)).

\begin{figure}
  \centering
  \includegraphics{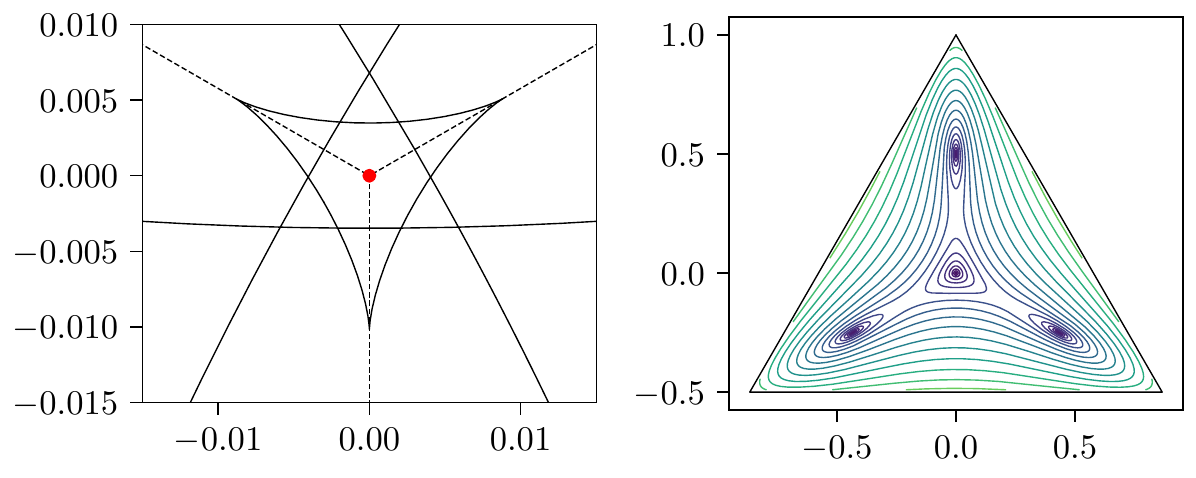}
  \caption{The left plot shows the Maxwell set (dashed lines) beyond
    Ellis-Wang together with the bifurcation set at the Ellis-Wang
    \(\beta = 4\log 2\).  The Maxwell set looks the same at lower
    temperatures.  The potential shows four equally deep minima for
    \(\alpha\) at the red dot.  Further decreasing the temperature,
    the central minimum will raise and eventually become a maximum.}
  \label{fig:maxwell-beyond}
\end{figure}

\FloatBarrier
\clearpage
\printbibliography
\end{document}